\documentclass{amsart}

\usepackage[utf8]{inputenc}
\usepackage[T1]{fontenc}
\usepackage{amssymb}
\usepackage{mathtools}
\usepackage{booktabs}
\usepackage{multirow}
\usepackage{graphicx}
\usepackage{xcolor}
\usepackage{tikz}
\usepackage[ruled,vlined,algo2e]{algorithm2e}
\usepackage[protrusion=true,expansion=false]{microtype}
\usepackage[round,authoryear]{natbib}
\usepackage{hyperref}
\hypersetup{
	colorlinks=true,
	linkcolor=blue,
	filecolor=magenta,
	urlcolor=cyan,
	citecolor=blue
}
\usepackage[capitalize]{cleveref}
\usepackage[text={440pt,575pt},headheight=9pt,centering]{geometry}

\usepackage{aliascnt}

\theoremstyle{plain}
\newtheorem{theorem}{Theorem}[section]

\newaliascnt{proposition}{theorem}

\aliascntresetthe{proposition}

\newaliascnt{lemma}{theorem}
\newtheorem{lemma}[lemma]{Lemma}
\aliascntresetthe{lemma}

\newaliascnt{corollary}{theorem}
\newtheorem{corollary}[corollary]{Corollary}
\aliascntresetthe{corollary}

\theoremstyle{definition}
\newaliascnt{definition}{theorem}

\aliascntresetthe{definition}

\newaliascnt{example}{theorem}

\aliascntresetthe{example}

\newaliascnt{remark}{theorem}

\aliascntresetthe{remark}

\newcommand{\RR}{\ensuremath{\mathbb{R}}}
\newcommand{\PP}{\ensuremath{\mathbb{P}}}
\newcommand{\one}{\ensuremath{\mathbf{1}}}

\newcommand{\cE}{\ensuremath{\mathcal{E}}}

\newcommand{\cD}{\ensuremath{\mathcal{D}}}
\newcommand{\CM}{\ensuremath{\operatorname{CM}}}
\newcommand{\spa}{\ensuremath{\mathrm{span}}}
\newcommand{\diag}{\ensuremath{\mathrm{diag}}}
\newcommand{\Det}{\ensuremath{\mathrm{Det}}}
\newcommand{\tr}{\ensuremath{\operatorname{tr}}}
\newcommand{\supp}{\ensuremath{\operatorname{supp}}}
\newcommand{\EMTPtwo}{\ensuremath{\mathrm{EMTP}_2}}
\newcommand{\MTPtwo}{\ensuremath{\mathrm{MTP}_2}}

\newcommand{\PSD}{\ensuremath{\succeq}}
\newcommand{\NSD}{\ensuremath{\preceq}}

\crefname{theorem}{theorem}{theorems}
\Crefname{theorem}{Theorem}{Theorems}
\crefname{lemma}{lemma}{lemmas}
\Crefname{lemma}{Lemma}{Lemmas}
\crefname{proposition}{proposition}{propositions}
\Crefname{proposition}{Proposition}{Propositions}
\crefname{definition}{definition}{definitions}
\Crefname{definition}{Definition}{Definitions}
\crefname{remark}{remark}{remarks}
\Crefname{remark}{Remark}{Remarks}
\crefname{corollary}{corollary}{corollaries}
\Crefname{corollary}{Corollary}{Corollaries}
\crefname{example}{example}{examples}
\Crefname{example}{Example}{Examples}
\crefname{algocf}{algorithm}{algorithms}
\Crefname{algocf}{Algorithm}{Algorithms}

\title[Spectral Sparsification of LCGGM and H\"{u}sler--Reiss Graphical Models]{Spectral
       Sparsification of Laplacian-Constrained Gaussian and
       H\"{u}sler--Reiss Graphical Models}

\author{Ignacio {Echave-Sustaeta Rodríguez}}
\address{Department of Mathematics and Computer Science, Eindhoven University of Technology, Eindhoven, The Netherlands}
\email{i.echave.sustaeta.rodriguez@tue.nl}

\author{Aida Abiad}
\address{Department of Mathematics and Computer Science, Eindhoven University of Technology, Eindhoven, The Netherlands}
\email{a.abiad@tue.nl}

\author{Frank Röttger}
\address{Department of Applied Mathematics, University of Twente, Enschede, The Netherlands}
\email{f.rottger@utwente.nl}

\begin{document}

\begin{abstract}
Graph Laplacians encode graph structures in matrix form, and thus facilitate the application of linear algebra to graph theory. 
In statistics, two related families of probabilistic graphical models can be parameterized by graph Laplacians. The first one is the Laplacian-constrained Gaussian graphical model (LCGGM), which
imposes that the (pseudo-)inverse covariance matrix of a Gaussian random vector is a Laplacian matrix. Applications include graph signal processing and network topology learning.  
The second one is the H\"{u}sler--Reiss graphical model, which is considered as an extremal analog of the Gaussian graphical model, and can be used in extremal dependence modeling of floods, heatwaves,
and financial losses.  
For both models, the restriction to positive edge weights in the graph Laplacian gives rise to an approach for graph structure learning that does not require tuning parameters.
While these approaches yield a strong model fit in many settings, the resulting graph estimates are typically much denser than the underlying ground truth, limiting interpretability and scalability.
In order to improve the accuracy of Laplacian-constrained graph learning, we propose to use
\emph{spectral graph sparsification} as a post-estimation operation.  To do so, we replace the original Laplacian estimate by a sparser Laplacian that is spectrally close, and re-fit the model on the resulting graph.  
We refer to the two resulting methods as
\emph{Spectral-LCGGM} and \emph{Spectral-HR}.  
We investigate the properties of the proposed estimators and show several theoretical results on their performance.
Furthermore, we demonstrate that the newly proposed methods perform well by running simulations on
Erd\H{o}s--R\'{e}nyi and stochastic block model graphs, and we also showcase their applications
to real data.
\end{abstract}

\maketitle

\section{Introduction}\label{sec:intro}

Probabilistic graphical models represent multivariate dependence through graphs, with
edges connecting pairs of variables that interact directly after
accounting for the others \citep{lauritzen1996}.  Two recent families
of graphical models share the algebraic feature that their precision
matrix is a connected weighted graph Laplacian.
The first is the
\emph{Laplacian-constrained Gaussian graphical model} (LCGGM),
that is a degenerate Gaussian supported on the hyperplane $\one^\perp$, in which the Laplacian
structure is imposed as a modeling constraint \citep{Egilmez2017,Ying2020,WangZhaoFeng2022,CaiCardosoPalomarYing2023}.  LCGGMs are widely
used in graph signal processing and network topology learning, and are
a degenerate-rank counterpart of Gaussian graphical models under multivariate total positivity of order 2 (\MTPtwo{})
\citep{SlawskiHein2015,LauritzenUhlerZwiernik2019,WangRoyUhler2020}.  The second is the
\emph{H\"{u}sler--Reiss graphical model} (HRGM) under \emph{extremal \MTPtwo} (\EMTPtwo{}) \citep{engelkehitz,REZ2023}, which can be considered as an extremal analog of Gaussian graphical models (under \MTPtwo{}).  
Here, the Laplacian structure of the precision
characterizes the positivity assumption.

While both models perform well in settings where the graph Laplacian assumption is reasonable, standard estimators return overly dense graphs.
In fact, \citet{REZ2023} show that their estimator for H\"usler--Reiss graphical models under \EMTPtwo{} asymptotically identifies a super-graph of the underlying ground truth.
In addition to obscuring model identification, retaining spurious edges in the graph estimates can slow down downstream graph-level
computations.  Penalized alternatives, for example $\ell_1$- and
nonconvex-based formulations for LCGGMs \citep{Ying2020} or
the eglearn estimator for HRGMs
\citep{ELV2025}, can improve the sparsity of the graph estimates, but may incur a
measurable cost in fit quality.

In this paper, we propose to employ techniques from spectral graph theory in order to maintain the fit quality of the unpenalized estimators, while improving upon their sparsity.
Our approach is similar to the post-estimation spectral-sparsification approach
of \citet{SpectralMTP2026}, developed there for the full-rank Gaussian
\MTPtwo{} model.  Because the precision matrices for LCGGMs and HRGMs are already graph Laplacians, we can apply the deterministic linear-size sparsifier of
\citet{BatsonSpielmanSrivastava2012} (BSS).  
Our approach is the same in both settings: We learn a
dense graph estimate, sparsify the resulting precision matrix, and refit the models on the
sparsified support.  We call the resulting estimators
\emph{Spectral-LCGGM} and \emph{Spectral-HR}, respectively.
As a theoretical contribution, we bound the change in
log-likelihood induced by the sparsifier in both settings, and
establish model-class preservation, together with additional guarantees on variograms and the H\"usler--Reiss exponent measure
density (\Cref{sec:hr}).  The analysis rests on a single identity on
$\one^\perp$ that combines a Bregman representation of the
log-likelihood gap with the spectral approximation produced by the
sparsifier, a degenerate-rank analog of the identity used in
\citet{SpectralMTP2026}.
LCGGM simulations on
Erd\H{o}s--R\'{e}nyi and stochastic block-model ground truths, and
applications of Spectral-HR to extremal Danube river levels and extremal flight delays in the Southern United States 
show that the sparsifiers retain the log-likelihood of their denser baselines at a significantly reduced
edge count, with substantially better graph recovery.

\subsection*{Related work.}
Spectral sparsification was introduced by \citet{ST2004} and was implemented with various methods in \citet{SpSr2011,BatsonSpielmanSrivastava2012,KLP12},
being extended to sparsification of strictly diagonally dominant
$M$-matrices via approximate Gaussian elimination
\citep{KyngSachdeva2016}.  In statistics and machine learning, these
methods have been used for multiple applications, with one of the earliest being
\citet{Srivastava2011CovarianceEF}. One of the closest ideas
to ours is \citet{SadhanalaWangTibshirani2016}, in which these methods are 
used to sparsify graphs prior to Laplacian-regularized regression while preserving
statistical validity. Another application combines
ridge spectral sparsification with Laplacian learning for scalable
semi-supervised classification \citep{CalandrielloLazaricKoutisValko2018}.
The most direct precedent for the present paper is
\citet{SpectralMTP2026}, which introduces the post-estimation BSS
sparsifier and the Bregman--Loewner analysis for the full-rank
Gaussian \MTPtwo{} model.  The LCGGM and \EMTPtwo{} H\"{u}sler--Reiss
settings in the present paper share the same template but require degenerate-rank
versions of the bounds: Kullback--Leibler (KL) divergence on $\one^\perp$, pseudo-determinants, and
trace-norm residuals through the projector $P$.  A related LCGGM
precedent is \citet{WangZhaoFeng2022}, who learn a Laplacian-constrained graph
by greedy spectral densification: starting from a sparse base graph, they apply
spectral-sparsification-style criteria to incrementally add the most spectrally
critical edges.  We instead apply sparsification as a post-estimation step
with explicit Loewner-order, log-likelihood, and KL guarantees.  On the modeling side of LCGGMs,
common estimators include the unpenalized CGL of \citet{Egilmez2017},
the GLE-ADMM algorithm of \citet{KumarYingCardosoPalomar2020}, and
the nonconvex NGL-SCAD estimator of \citet{Ying2020}.  For HRGMs,
the canonical dense baseline is the
\EMTPtwo{}-constrained surrogate maximum likelihood estimator (MLE) \citep{REZ2023} and the canonical
sparse baseline is the eglearn estimator
\citep{ELV2025}.

\section{Preliminaries}\label{sec:preliminaries}

\subsection{Graphs, Laplacians, and Spectral Sparsification}\label{sec:graph_prelims}

We consider connected weighted graphs $G=(V,E,c)$ with vertex set $V$, edge set
$E$, and symmetric weight function $c:V\times V\to[0,\infty)$.
The weighted Laplacian is $L_G := \mathrm{Deg}(G)-A$, where
$A=(c_{ij})$ is the adjacency matrix and $\mathrm{Deg}(G)=\diag(k_i)$ the diagonal degree
matrix with $k_i=\sum_j c_{ij}$. For a connected weighted graph, $L_G$ is
positive semidefinite, satisfies $L_G\one=0$, and has rank $|V|-1$.

For $\varepsilon\in(0,1)$, a graph $G'$ on the same vertex set is an
\emph{$\varepsilon$-spectral approximation} of $G$ if
\begin{equation}\label{eq:spectral_approx_graph}
(1-\varepsilon)\,z^\top L_G\,z \;\le\; z^\top L_{G'} z \;\le\; (1+\varepsilon)\,z^\top L_G\,z
\qquad\forall\, z\in\RR^{|V|}.
\end{equation}

For symmetric matrices $A,B\in\RR^{|V|\times|V|}$ and a subspace
$W\subseteq\RR^{|V|}$, we write $A\NSD B$ \emph{on $W$} if
$z^\top A z\le z^\top B z$ for every $z\in W$.  When the suffix
``on $W$'' is omitted, we mean $W=\RR^{|V|}$.  $A\PSD B$ is defined
symmetrically.
This is the \emph{Loewner order}, a partial order on symmetric
matrices.  For Laplacians the natural choice is $W=\one^\perp$,
since their kernel $\spa(\one)$ contributes no information.
In this notation, \eqref{eq:spectral_approx_graph} reads
$(1-\varepsilon)L_G\NSD L_{G'}\NSD(1+\varepsilon)L_G$.

The following deterministic construction, due to
\citet{BatsonSpielmanSrivastava2012}, achieves such an approximation
with a linear-size number of edges. We use it throughout the paper.

\begin{theorem}\label{thm:ss}
Let $L$ be a connected weighted Laplacian on $d$ vertices and fix $\eta>1$.
There is a deterministic polynomial-time algorithm that outputs a reweighted
subgraph Laplacian $\widetilde L$ supported on at most $\lceil\eta(d-1)\rceil$
edges that satisfies
\[
L\NSD\widetilde L\NSD\kappa(\eta)\,L\qquad\text{on }\one^\perp,
\qquad\kappa(\eta) := \Bigl(\tfrac{\sqrt\eta+1}{\sqrt\eta-1}\Bigr)^2.
\]
After rescaling $\widetilde L$ by the constant $2/(1+\kappa(\eta))$, this
becomes the symmetric $(1\pm\varepsilon)$ Loewner bound
$(1-\varepsilon)L\NSD\widetilde L\NSD(1+\varepsilon)L$ on $\one^\perp$ with
$\varepsilon=(\kappa(\eta)-1)/(\kappa(\eta)+1)=2\sqrt\eta/(\eta+1)$.
\end{theorem}

\subsection{Laplacian-Constrained Gaussian Graphical Models}\label{sec:lcggm_prelims}
A \emph{Laplacian-constrained Gaussian graphical model} (LCGGM) is a
centered, degenerate Gaussian random vector $X$ with support $\one^\perp\subset\RR^d$,
whose precision matrix $K$ is a connected weighted graph Laplacian
\citep{Egilmez2017,Ying2020,WangZhaoFeng2022}.  The model has a probability density with respect to the Lebesgue measure on $\one^\perp$:
\begin{equation}\label{eq:lcggm_density}
f_K(x)=(2\pi)^{-(d-1)/2}\,\Det(K)^{1/2}
       \exp\!\bigl(-\tfrac12\,x^\top K x\bigr),
\qquad x\in\one^\perp.
\end{equation}
Here, $\Det$ denotes the pseudo-determinant, that is the product of the nonzero eigenvalues.  For a sample covariance matrix $S\PSD 0$ obtained from $n$ i.i.d.\ copies of $X$, the LCGGM
log-likelihood up to additive constants is equal to
\begin{equation}\label{eq:lcggm_loglik}
\ell(K;S):=\log\Det(K)-\tr(KS).
\end{equation}
This gives rise to the LCGGM-MLE given $S$
\begin{equation}\label{eq:LCGGM-MLE}
    \widehat{K}=\arg\max_K \ell(K;S)\quad \text{s.t. } K_{ij} \le 0, \;\; \forall i\neq j.
\end{equation}
As $c$ is restricted to $[0,\infty)$, LCGGM is the Laplacian analog of the Gaussian \MTPtwo{} distributions
\citep{LauritzenUhlerZwiernik2019,WangRoyUhler2020}: both impose
nonpositive off-diagonal entries on the precision matrix. However, LCGGM
additionally requires zero row sums, so $K$ is the Laplacian of a positively
weighted graph. 

\subsection{H\"{u}sler--Reiss Graphical Models and graph Laplacians}\label{sec:hr_prelims}

Let $X=(X_1,\ldots,X_d)$ be a random vector with continuous marginal cumulative distribution functions $F_i$.
After transforming to exponential margins $X_i^*:=-\log(1-F_i(X_i))$,
the threshold-exceedance limit
\begin{equation}\label{eq:MPD}
\PP(Y\le y)
= \lim_{u\to\infty}
\PP\bigl(X^*-u\one\le y\mid\max_{i\in[d]}X_i^*\ge u\bigr)
\end{equation}
yields, under multivariate regular variation, a multivariate (generalized) Pareto vector $Y$.
Assuming asymptotic dependence, $Y$ is supported on
$\cE:=\{y\in\RR^d:\max_i y_i>0\}$
\citep{roo2006,HES2022}.

A $d$-variate H\"{u}sler--Reiss distribution is a multivariate Pareto distribution that can be parametrized by a conditionally negative definite variogram
matrix
\[
\Gamma\in\cD^d :=
\bigl\{\Gamma\in\RR^{d\times d}:\Gamma=\Gamma^\top,\,\diag(\Gamma)=0,\,
x^\top\Gamma x<0\;\forall\,x\neq 0,\,x\perp\one\bigr\}.
\]
Its probability density is proportional to a function $\lambda$, typically referred to as exponent measure density:
\begin{equation}\label{eq:hr_density_gamma}
f_\Gamma(y)\propto \lambda_\Gamma(y)=\frac{1}{\sqrt{-\det\CM(2\pi\Gamma)}}
\exp\!\Bigl(-\tfrac12
\begin{pmatrix}y&1\end{pmatrix}
\CM(\Gamma)^{-1}
\begin{pmatrix}y\\1\end{pmatrix}\Bigr).
\end{equation}
Here, $\CM(\Gamma):=\bigl(\begin{smallmatrix}-\Gamma/2&\one\\\one^\top&0\end{smallmatrix}\bigr)$
is the Cayley--Menger matrix \citep{devriendt2022a,devriendt2026extremal}.
Let
$P:=I-\frac{1}{d}\one\one^\top$ be a projection matrix onto $\one^\perp$, and define 
\begin{equation}\label{eq:Sigmahat_def}
\Sigma\;:=\;-\tfrac{1}{2}\,P\,\Gamma\,P.
\end{equation}
The variogram is recovered from $\Sigma$ via
$\Gamma=\diag(\Sigma)\one^\top+\one\diag(\Sigma)^\top-2\Sigma$.
If $\Gamma\in\cD^d$, then $\Sigma$ is positive semidefinite with kernel $\spa(\one)$ and has a Moore--Penrose pseudoinverse $\Theta=\Sigma^+$.
We write $\Gamma_\Theta$ for the variogram associated with a precision $\Theta$ via this identity (with $\Sigma=\Theta^+$).
The matrix $\Theta$ allows for a parametric encoding of HRGMs:
$\Theta_{ij}=0$ is equivalent to the extremal conditional
independence of $Y_i$ and $Y_j$ given the remaining coordinates
\citep{HES2022}.
Furthermore, for H\"usler--Reiss distributions, the restriction for $\Theta$ to be a graph Laplacian matrix ($\Theta_{ij}\le 0$ for all $i\neq j$) is equivalent to an extremal notion of total positivity (\EMTPtwo{}), compare \citet{REZ2023}.

Assuming i.i.d.~data for some random vector $X$ in the domain of attraction of a H\"usler--Reiss random vector $Y$ as in \eqref{eq:MPD},
the empirical variogram estimator $\overline\Gamma$ of \citet{EV2022} is a consistent estimator of $\Gamma$, which can serve as a summary statistic for parameter estimation in HRGMs \citep{HES2022}.
In order to estimate $\Theta$ for HRGMs under the constraint of \EMTPtwo{}, \citet{REZ2023} propose a surrogate MLE under graph Laplacian constraint via the LCGGM log-likelihood with sample covariance $\overline{\Sigma}=-\tfrac{1}{2}\,P\,\overline{\Gamma}\,P $, that is
\begin{equation}\label{eq:HR-sMLE}
    \widehat{\Theta}=\arg\max_\Theta \ell(\Theta;\overline{\Sigma})\quad \text{s.t. } \Theta_{ij} \le 0, \;\; \forall i\neq j.
\end{equation}
Note that we can use this relationship between both estimation procedures to apply results and algorithms from LCGGM to H\"usler--Reiss.

\section{Laplacian-Constrained Gaussian Graphical Models}\label{sec:lcggm}

\subsection{Theory}\label{sec:lcggm_theory}

The LCGGM-MLE $\widehat K$ obtained from a sample covariance $S$ via
\eqref{eq:LCGGM-MLE} is by construction a connected graph Laplacian,
so \Cref{thm:ss} applies directly: there exists a sparsifier
$\widetilde K$ that is also a connected graph Laplacian and satisfies
the spectral approximation
\begin{equation}\label{eq:spec_lcggm}
(1-\varepsilon)\widehat{K}\NSD\widetilde K\NSD(1+\varepsilon)\widehat{K}
\qquad\text{on }\one^\perp.
\end{equation}
We use the notation
\begin{equation}\label{eq:kl_perp}
D_{\mathrm{KL}}^\perp(\widehat{K}\,\|\,\widetilde K)
\;:=\;\tfrac12\bigl[\tr(\widetilde K\,\widehat{K}^+)\;-\;(d-1)\;-\;\log\Det(\widetilde K\,\widehat{K}^+)\bigr]
\end{equation}
for the KL divergence between the two centered degenerate Gaussians
on $\one^\perp$ with precisions $\widehat{K}$ and $\widetilde K$. See \Cref{lem:kl_perp_identity} for the derivation.

Our goal is to control the LCGGM log-likelihood gap
$\ell(\widetilde K;T)-\ell(\widehat K;T)$ induced by replacing
$\widehat K$ by its sparsifier $\widetilde K$, for any (projected)
sample covariance matrix $T$ supported on $\one^\perp$
(i.e., $T\one=0$).
The theorem below decomposes this gap into a structural cost
(the KL term, which is the optimization-side price of moving away
from $\widehat K$) and a data residual (a trace term that vanishes
when $T = \widehat K^+$), and bounds it linearly in $\varepsilon$ in
both directions.  Two regimes are worth distinguishing.  If $T = S$
(in-sample evaluation), then $\widehat K$ is the constrained MLE on
$S$ over the cone of connected Laplacians, so by definition no
feasible $\widetilde K$ can improve the in-sample log-likelihood and
the gap is non-positive.  If instead $T \neq S$, for instance when
$T$ is the sample covariance of held-out test data, the data
residual
$\tr\!\bigl((\widetilde K-\widehat K)(T-\widehat K^+)\bigr)$ can take
either sign, and the sparsifier may improve the log-likelihood on
$T$ relative to $\widehat K$.

\begin{theorem}\label{thm:main_lcggm}
Let $\varepsilon\in(0,1)$, let $\widehat{K},\widetilde K$ be connected
graph Laplacians satisfying \eqref{eq:spec_lcggm}, and let $T\PSD 0$
be symmetric with $T\one=0$.  Then
\begin{equation}\label{eq:bregman_lcggm}
\ell(\widetilde K;T)-\ell(\widehat{K};T)
=-2\,D_{\mathrm{KL}}^\perp(\widehat{K}\,\|\,\widetilde K)
-\tr\!\bigl((\widetilde K-\widehat{K})(T-\widehat{K}^+)\bigr),
\end{equation}
and, setting
$R(T):=\bigl\|\widehat{K}^{1/2}(T-\widehat{K}^+)\widehat{K}^{1/2}\bigr\|_*$
(with $\widehat{K}^{1/2}$ the symmetric square root of $\widehat{K}$
on $\one^\perp$, zero on $\spa(\one)$, and $\|\cdot\|_*$ the trace
norm),
\begin{equation}\label{eq:loglik_lcggm_upper}
\ell(\widetilde K;T)-\ell(\widehat{K};T)\;\le\;\varepsilon\,R(T),
\end{equation}
\begin{equation}\label{eq:loglik_lcggm_bound}
-\frac{(d-1)\varepsilon^2}{2(1-\varepsilon)}-\varepsilon\,R(T)
\;\le\;\ell(\widetilde K;T)-\ell(\widehat{K};T)
\;\le\;\varepsilon\,R(T).
\end{equation}
\end{theorem}

Beyond the log-likelihood bound, any sparsifier $\widetilde K$
satisfying \eqref{eq:spec_lcggm} is automatically a valid LCGGM
precision matrix, with algebraic connectivity (the smallest nonzero
eigenvalue $\lambda_2$, also called the Fiedler eigenvalue) preserved
within $1\pm\varepsilon$.

\Cref{thm:main_lcggm} guarantees that any spectral sparsifier of
$\widehat K$ loses only an $O(\varepsilon)$ amount of in-sample
log-likelihood.  This motivates the three-step pipeline of
\Cref{alg:spectral_mle}: (1) compute the dense LCGGM-MLE $\widehat K$
\eqref{eq:LCGGM-MLE} on the sample covariance $S$, (2) apply the BSS
sparsifier of \Cref{thm:ss} to $\widehat K$, and (3) refit the
LCGGM-MLE restricted to the edges retained by BSS.

The same algorithm applies to the H\"{u}sler--Reiss setting of
\Cref{sec:hr}: the LCGGM-MLE \eqref{eq:LCGGM-MLE} and the
\EMTPtwo{} surrogate HR-MLE \eqref{eq:HR-sMLE} solve the same convex
program $\arg\max\{\ell(K;D):K\in\mathcal{C}\}$ over the cone
$\mathcal{C}$ of connected graph Laplacians, with only the input $D$
changing ($S$ for LCGGM, the data summary $\overline\Sigma$ for HR).
\Cref{alg:spectral_mle} is stated with $D$ as input.

\begin{algorithm2e}[ht]
\DontPrintSemicolon
\caption{Spectral-LCGGM and Spectral-HR (unified).}
\label{alg:spectral_mle}
\KwIn{Data summary $D\in\RR^{d\times d}$ ($S$ for LCGGM,
$\overline\Sigma$ for HR); BSS parameter $\eta>1$.}
\KwOut{Sparse precision Laplacian $\widetilde L^{\mathrm{refit}}$.}
\textbf{Step 1: Dense MLE:}
$L\gets\arg\max\bigl\{\ell(K;D):K\in\mathcal{C}\bigr\}$, with
$\mathcal{C}$ the cone of connected graph Laplacians\;
\textbf{Step 2: BSS sparsification:} run \Cref{thm:ss} on $L$ at
$\eta$, obtaining $\widetilde L$ on $\le\lceil\eta(d-1)\rceil$ edges
satisfying $(1-\varepsilon)L\NSD\widetilde L\NSD(1+\varepsilon)L$ on
$\one^\perp$, with $\varepsilon=(\kappa(\eta)-1)/(\kappa(\eta)+1)$;
set $\widetilde E\gets\{(i,j):\widetilde L_{ij}\neq 0,\,i\neq j\}$\;
\textbf{Step 3: Refit:} \Return $\widetilde L^{\mathrm{refit}}\gets
\arg\max\bigl\{\ell(K;D):K\in\mathcal{C},\,\supp(K)\subseteq\widetilde E\bigr\}$\;
\end{algorithm2e}

Tuning proceeds through $\eta$, selected by BIC on the training data;
smaller $\eta$ yields sparser subgraphs at looser spectral certificates.

\subsection{Experiments}\label{sec:lcggm_experiments}

We evaluate Spectral-LCGGM on simulated data drawn from an LCGGM whose
ground-truth precision $K^\star$ is the Laplacian of a weighted random
graph (the specific graph families and dimensions are detailed below),
so $E^\star$ denotes the edge set of that random graph.  For each
method, we report the edge count $|\widehat E|$ (number of
off-diagonals of the estimate exceeding $10^{-3}$ in magnitude), the
precision $|\widehat E\cap E^\star|/|\widehat E|$, the recall
$|\widehat E\cap E^\star|/|E^\star|$, and the F1 score (harmonic mean
of precision and recall), together with the LCGGM test log-likelihood
$\ell(\widehat K;S_{\mathrm{test}})$ on a held-out sample covariance.
Tuning parameters are selected on the training split by the Bayesian information criterion (BIC) 
\begin{equation}\label{eq:bic}
\mathrm{BIC}(M)\;:=\;-n_{\mathrm{train}}\,\ell(M;D_{\mathrm{train}})\;+\;k(M)\,\log n_{\mathrm{train}},
\end{equation}
with $\ell$ the LCGGM log-likelihood \eqref{eq:lcggm_loglik},
$n_{\mathrm{train}}$ the number of training data points, $D_{\mathrm{train}}$ the training covariance in each setting, and $k(M)$ the edge count of the graph given by the Laplacian $M$.

We compare Spectral-LCGGM (\Cref{alg:spectral_mle}) against two
baselines that span the sparsity--fit trade-off: CGL
\citep{Egilmez2017}, the unpenalized LCGGM-MLE 
and NGL-SCAD \citep{Ying2020}, a nonconvex SCAD-penalized LCGGM
solver.  The Graphical Lasso is excluded because $\ell_1$-norm
penalization is known not to produce sparse estimates in the LCGGM
setting \citep{Ying2020}.

The ground-truth precision matrix $K^\star$ is the Laplacian of a
weighted random graph, with edge weights drawn from $\mathrm{Unif}[0.5,
1.5]$ to avoid eigenvalue multiplicities.  We consider two random graph
families: Erd\H{o}s--R\'{e}nyi (ER) graphs with uniform
degree, and stochastic block models (SBM) with five blocks
($p_{\mathrm{in}}=0.12$, $p_{\mathrm{out}}=0.01$) inducing community
structure.  Each topology is realized at two dimensions,
$d\in\{100,200\}$, with $n_{\mathrm{train}}=4d$ training and
$n_{\mathrm{test}}=4{,}000$ test observations, over $B=10$
replications.

\Cref{tab:sim_lcggm} reports the BIC-selected results.  Across all
four settings the unpenalized LCGGM-MLE (CGL) recovers a graph much
denser than the truth, with high recall and very low precision.
NGL-SCAD (BIC-selected $\alpha$) brings the edge count closer to
the truth with substantially improved F1.  Spectral-LCGGM
consistently delivers edge counts within $5\%$ of the truth and
achieves the highest F1 ($\ge 0.94$) in every setting, while also
obtaining the best test log-likelihood.  The test log-likelihood
gap to NGL-SCAD is narrow, reflecting that NGL-SCAD also fits the
data well; Spectral-LCGGM's advantage materialises primarily on
graph recovery.

\begin{table}[ht]
\centering
\resizebox{\linewidth}{!}{%
\begin{tabular}{llrrrrr}
\toprule
Setting & Method & Edges & Precision & Recall & F1 & Test log-lik.\ \\
\midrule
\multirow{3}{*}{\shortstack[l]{ER, $d=100$\\$|E^\star|=176$}}
  & CGL                     & $\phantom{1{,}}506.3\,(4.9)$ & $0.348\,(0.003)$ & $\mathbf{1.000}\,(0.000)$ & $0.516\,(0.004)$ & $-8{,}173\,(125)$ \\
  & NGL-SCAD                & $\phantom{1{,}}240.7\,(2.8)$ & $0.731\,(0.008)$ & $0.999\,(0.001)$ & $0.844\,(0.005)$ & $-7{,}269\,(123)$ \\
  & Spectral-LCGGM   & $\mathbf{185.0}\,(1.0)$ & $\mathbf{0.950}\,(0.006)$ & $0.998\,(0.001)$ & $\mathbf{0.973}\,(0.003)$ & $\mathbf{-6{,}631}\,(135)$ \\
\midrule
\multirow{3}{*}{\shortstack[l]{SBM, $d=100$\\$|E^\star|=171$}}
  & CGL                     & $\phantom{1{,}}453.9\,(4.9)$ & $0.377\,(0.004)$ & $\mathbf{1.000}\,(0.000)$ & $0.548\,(0.004)$ & $-27{,}284\,(141)$ \\
  & NGL-SCAD                & $\phantom{1{,}}233.7\,(2.0)$ & $0.731\,(0.007)$ & $0.998\,(0.001)$ & $0.844\,(0.005)$ & $-26{,}378\,(146)$ \\
  & Spectral-LCGGM   & $\mathbf{176.6}\,(0.7)$ & $\mathbf{0.963}\,(0.004)$ & $0.995\,(0.001)$ & $\mathbf{0.979}\,(0.002)$ & $\mathbf{-25{,}952}\,(141)$ \\
\midrule
\multirow{3}{*}{\shortstack[l]{ER, $d=200$\\$|E^\star|=391$}}
  & CGL                     & $1{,}259.5\,(8.4)$ & $0.311\,(0.002)$ & $\mathbf{1.000}\,(0.000)$ & $0.474\,(0.002)$ & $19{,}469\,(253)$ \\
  & NGL-SCAD                & $\phantom{1{,}}488.6\,(9.1)$ & $0.801\,(0.014)$ & $0.998\,(0.000)$ & $0.888\,(0.009)$ & $21{,}085\,(269)$ \\
  & Spectral-LCGGM   & $\mathbf{392.9}\,(0.5)$ & $\mathbf{0.992}\,(0.001)$ & $0.997\,(0.001)$ & $\mathbf{0.994}\,(0.001)$ & $\mathbf{21{,}889}\,(229)$ \\
\midrule
\multirow{3}{*}{\shortstack[l]{SBM, $d=200$\\$|E^\star|=623$}}
  & CGL                     & $1{,}886.9\,(16.8)$ & $0.329\,(0.003)$ & $\mathbf{0.995}\,(0.001)$ & $0.494\,(0.003)$ & $248{,}258\,(242)$ \\
  & NGL-SCAD                & $\phantom{1{,}}781.4\,(1.8)$ & $0.770\,(0.002)$ & $0.966\,(0.002)$ & $0.857\,(0.002)$ & $249{,}607\,(234)$ \\
  & Spectral-LCGGM   & $\mathbf{632.8}\,(1.9)$ & $\mathbf{0.936}\,(0.003)$ & $0.950\,(0.002)$ & $\mathbf{0.943}\,(0.002)$ & $\mathbf{250{,}223}\,(246)$ \\
\bottomrule
\end{tabular}}
\caption{LCGGM simulation across two topologies and two dimensions
($d\in\{100,200\}$). The two topologies are Erd\H{o}s--R\'{e}nyi (ER)
with target density $2(d-1)$ edges, and stochastic block model (SBM)
with five equal-sized blocks, $p_{\mathrm{in}}=0.12$,
$p_{\mathrm{out}}=0.01$ (so the true edge count $|E^\star|$ grows with
$d$).  Means over $B=10$ replications with standard errors in
parentheses.  Bold: best on each metric within each $(d,
\text{topology})$ block.}
\label{tab:sim_lcggm}
\end{table}

\section{H\"{u}sler--Reiss Graphical Models}\label{sec:hr}

\subsection{Theory}\label{sec:hr_theory}

Under \EMTPtwo{}, the precision matrix $\Theta$ is a
graph Laplacian (\Cref{sec:hr_prelims}), so \Cref{thm:ss} applies
directly to $\Theta$.  Moreover, by \eqref{eq:HR-sMLE} the surrogate
H\"{u}sler--Reiss MLE $\widehat\Theta$ is the LCGGM-MLE
\eqref{eq:LCGGM-MLE} evaluated at $\overline\Sigma$ in
place of the sample covariance.  No separate analysis is therefore
required for the H\"{u}sler--Reiss case: the LCGGM theory of
\Cref{sec:lcggm_theory} transfers verbatim under the substitution
$\widehat{K}\to\widehat \Theta$, $\widetilde K\to\widetilde\Theta$, $S\to\overline\Sigma$,
$\widehat K^+\to \widehat \Sigma=\widehat\Theta^+$, with the spectral approximation
\begin{equation}\label{eq:spec_hr}
(1-\varepsilon)\widehat \Theta\NSD\widetilde\Theta\NSD(1+\varepsilon)\widehat \Theta
\qquad\text{on }\one^\perp
\end{equation}
in place of \eqref{eq:spec_lcggm}.  The unified estimator of
\Cref{alg:spectral_mle}, run with $\widehat\Theta$ and $D=\overline\Sigma$,
produces Spectral-HR.
Concretely, \Cref{thm:main_lcggm} gives
\begin{equation}\label{eq:bregman_hr}
\ell(\widetilde\Theta;T)-\ell(\widehat{\Theta};T)
=-2\,D_{\mathrm{KL}}^\perp(\widehat{\Theta}\,\|\,\widetilde\Theta)
-\tr\!\bigl((\widetilde\Theta-\widehat{\Theta})(T-\widehat{\Sigma})\bigr)
\end{equation}
and, with $\gamma(T)=\diag(T)\one^\top+\one\diag(T)^\top-2T$ we have for $R_{\widehat{\Gamma}}:=\bigl\|\widehat{\Theta}^{1/2}(T-\widehat{\Sigma})\widehat{\Theta}^{1/2}\bigr\|_*
=\tfrac12\bigl\|\widehat{\Theta}^{1/2}P(\gamma(T)-\widehat{\Gamma})P\,\Theta^{1/2}\bigr\|_*$
(the subscript reflects the underlying variogram-estimation error), that
\begin{equation}\label{eq:loglik_hr_bound}
-\frac{(d-1)\varepsilon^2}{2(1-\varepsilon)}-\varepsilon\,R_{\widehat{\Gamma}}
\;\le\;\ell(\widetilde\Theta;\overline\Sigma)-\ell(\widehat{\Theta};\overline\Sigma)
\;\le\;\varepsilon\,R_{\widehat{\Gamma}}.
\end{equation}
Under any consistent variogram estimator, $R_{\widehat{\Gamma}}\to 0$ in
probability, so the residual term in \eqref{eq:loglik_hr_bound}
vanishes asymptotically.  Any sparsifier $\widetilde\Theta$
satisfying \eqref{eq:spec_hr} is itself a connected graph Laplacian
(by construction of BSS), hence the precision of an \EMTPtwo{} H\"usler--Reiss
distribution, with algebraic connectivity $\lambda_2$ preserved
within $1\pm\varepsilon$.

The remainder of this subsection collects two guarantees with relevance to
the H\"{u}sler--Reiss case: all variogram entries $\Gamma_{ij}$ are
preserved multiplicatively, and the H\"usler--Reiss logarithmic
exponent measure density $\log\lambda(y)$ is controlled pointwise.

\begin{theorem}\label{thm:variogram_preservation}
Let $\Theta,\widetilde\Theta$ be connected graph Laplacians satisfying
\eqref{eq:spec_hr}, with associated variograms $\Gamma,\widetilde\Gamma$. For
every pair $i\neq j$,
\begin{equation}\label{eq:variogram_preservation}
\frac{\Gamma_{ij}}{1+\varepsilon}\;\le\;\widetilde\Gamma_{ij}\;\le\;\frac{\Gamma_{ij}}{1-\varepsilon}.
\end{equation}
In other words, the variogram is preserved
multiplicatively within a factor $1/(1\pm\varepsilon)$,
independently of $d$ and $\lambda_2(\Theta)$.
\end{theorem}

The multiplicative variogram control transfers, via the Cayley--Menger
representation \eqref{eq:hr_density_gamma}, to a pointwise bound on the
H\"usler--Reiss logarithmic exponent measure density.  Before stating
it, define
\begin{equation}\label{eq:M_C_def}
M_\Theta\;:=\;\bigl\|\CM(\Gamma)^{-1}\bigr\|_{\mathrm{op}}
\;\le\;\|\Theta\|_{\mathrm{op}}+\|p\|_2+|\sigma^2|,
\end{equation}
where $\Theta\in\RR^{d\times d}$, $p\in\RR^d$, and $\sigma^2\in\RR$
are the blocks of
$\CM(\Gamma)^{-1}=\bigl(\begin{smallmatrix}\Theta & p\\ p^\top & \sigma^2\end{smallmatrix}\bigr)$,
with $\Theta=\Sigma^+$ in the top-left $d\times d$ block (Fiedler--Bapat
identity, see \cite{devriendt2022a}). Set also $\eta_C:=\varepsilon\,\|\Gamma\|_2/(2(1-\varepsilon))$.

\begin{corollary}\label{cor:loglambda_pointwise}
Let $\Theta,\widetilde\Theta$ be connected graph Laplacians satisfying
\eqref{eq:spec_hr}, with associated variograms $\Gamma,\widetilde\Gamma$,
and suppose $\eta_C M_\Theta<1$.  Then for every $y\in\cE$,
\begin{equation}\label{eq:loglambda_pointwise}
\bigl|\log\lambda_{\widetilde\Gamma}(y)-\log\lambda_\Gamma(y)\bigr|
\;\le\;\frac{d+1}{2}\cdot\frac{\eta_C M_\Theta}{1-\eta_C M_\Theta}
\;+\;\frac{\|y\|_2^2+1}{2}\cdot\frac{\eta_C M_\Theta^2}{1-\eta_C M_\Theta},
\end{equation}
and for any sample $Y_1,\ldots,Y_n\in\cE$ with
$\overline{\|Y\|^2}:=n^{-1}\sum_k\|Y_k\|_2^2$,
\begin{equation}\label{eq:loglambda_sample}
\frac{1}{n}\Bigl|\sum_{k=1}^n\bigl[\log\lambda_{\widetilde\Gamma}(Y_k)-\log\lambda_\Gamma(Y_k)\bigr]\Bigr|
\;\le\;\frac{d+1}{2}\cdot\frac{\eta_C M_\Theta}{1-\eta_C M_\Theta}
\;+\;\frac{\overline{\|Y\|^2}+1}{2}\cdot\frac{\eta_C M_\Theta^2}{1-\eta_C M_\Theta}.
\end{equation}
\end{corollary}

Both bounds are $O(\varepsilon)$ as $\varepsilon\downarrow 0$ and
independent of the spectral gap $\lambda_2(\Theta)$.
The sample-average bound \eqref{eq:loglambda_sample} controls the
data-dependent part of the H\"usler--Reiss  log-likelihood:
writing $\ell^{\mathrm{HR}}(\Theta;Y_{1:n})
=\sum_k\log\lambda_\Gamma(Y_k)-n\log V(\Gamma)$, the remaining term
$\log V(\widetilde\Gamma)-\log V(\Gamma)$ is deterministic in
$\Gamma,\widetilde\Gamma$ and depends only on the variogram
perturbation $\widetilde\Gamma-\Gamma$, which is itself controlled
multiplicatively by \eqref{eq:variogram_preservation}.  We do not
develop an explicit bound on this term here.
The condition $\eta_C M_\Theta<1$ holds for any fixed $\Theta$ once
$\varepsilon$ is small enough.

\subsection{Experiments}\label{sec:hr_experiments}

We evaluate Spectral-HR against the \EMTPtwo{}-MLE \citep{REZ2023} and
the two majority-vote variants of eglearn \citep{ELV2025},
neighborhood-selection (eglearn-NS) and graphical-lasso (eglearn-glasso),
on two H\"{u}sler--Reiss real-data benchmark datasets. The upper Danube river-network discharge dataset
\citep{engelkehitz} has a modest sample size, which makes it difficult to split the data with a
high enough sample size (having in mind that we also have to threshold the data), so methods are compared by BIC and AIC on the full
thresholded sample.  The Texas Cluster of U.S.\ airports \citep{HES2022}
is a much larger dataset, so we can use the standard
train/test date split of \citet{HES2022} and held-out H\"usler--Reiss
log-likelihood as the headline metric.  In both experiments the
log-likelihood used for selection and held-out evaluation is the
actual H\"usler--Reiss log-likelihood
$\ell^{\mathrm{HR}}(\widehat\Theta):=\sum_k\log f_{\widehat\Gamma}(Y_k)$.

\subsection*{Upper Danube river-network discharges
\citep{engelkehitz}.}
This is a standard H\"{u}sler--Reiss benchmark.
The data are daily seasonal-residual discharges at $d=31$ gauging
stations along the upper Danube and its tributaries, with $n=428$
exceedance-cluster days as bundled in the \texttt{graphicalExtremes}
R package.  Since there is no canonical extremal ground-truth graph
for river discharges, we use the whole dataset both for fitting and for the
tuning criteria, and compare methods purely on in-sample model
selection.  At threshold $p=0.90$ the thresholded sample has
$n^{\mathrm{thr}}=117$ exceedances, on which we evaluate both
\begin{equation}\label{eq:bic_aic_hr}
\mathrm{BIC}(\widehat\Theta)
\;=\;-2\,\ell^{\mathrm{HR}}(\widehat\Theta)
\;+\;k(\widehat\Theta)\,\log n^{\mathrm{thr}},
\qquad
\mathrm{AIC}(\widehat\Theta)
\;=\;-2\,\ell^{\mathrm{HR}}(\widehat\Theta)
\;+\;2\,k(\widehat\Theta),
\end{equation}
where $k(\widehat\Theta)$ is the edge count of the graph given by $\widehat\Theta$.  Each tunable method is
selected twice: once by BIC, once by AIC.
The \EMTPtwo{}-MLE has no tuning parameter.  Tuning grids are
$\rho\in[0.005,0.5]$ for eglearn-NS, $\rho\in[0.005,0.75]$
($20$ equidistant points) for eglearn-glasso, and
$\eta\in[1.1,100]$ for Spectral-HR.

\Cref{tab:realdata_hr_danube} reports the BIC- and AIC-selected fits.
Spectral-HR gives the sparsest fits by a clear margin ($46$--$53$
edges, vs $99$--$130$ for eglearn-NS, $171$ for eglearn-glasso, and
$67$ for \EMTPtwo{}). eglearn-NS attains both the best BIC and the
best AIC, at much higher edge counts than Spectral-HR.  Against the
dense \EMTPtwo{} baseline, Spectral-HR improves both criteria
($\Delta\mathrm{BIC}=-71.6$, $\Delta\mathrm{AIC}=-28.6$) while using
roughly $70$\% of its edges; it trails the eglearn-NS winners on the
absolute criteria ($87$ BIC points / $302$ AIC points behind), trading
information-criterion fit for substantially sparser graphs.
\begin{table}[ht]
\centering
\begin{tabular}{llrrrr}
\toprule
Method & Selected by & Param & Edges & BIC & AIC \\
\midrule
\EMTPtwo{}                         & (none) & --        & $67$  & $1{,}336.07$           & $1{,}150.98$ \\
\midrule
\multirow{2}{*}{eglearn-NS}        & BIC    & $0.032$   & $99$  & $\mathbf{1{,}177.01}$  & $903.54$ \\
                                   & AIC    & $0.023$   & $130$ & $1{,}179.36$           & $\mathbf{820.30}$ \\
\midrule
eglearn-glasso                     & BIC \& AIC & $0.044$ & $171$ & $1{,}342.77$         & $870.42$ \\
\midrule
\multirow{2}{*}{Spectral-HR}       & BIC    & $4$       & $\mathbf{46}$ & $1{,}264.46$   & $1{,}137.40$ \\
                                   & AIC    & $25$      & $53$  & $1{,}268.75$           & $1{,}122.34$ \\
\bottomrule
\end{tabular}
\caption{Upper Danube river-network discharges
(H\"{u}sler--Reiss, $d=31$, $n=428$ days, threshold $p=0.90$,
$n^{\mathrm{thr}}=117$ exceedances).  Whole dataset used for
fitting and for both selection criteria \eqref{eq:bic_aic_hr}.
Each tunable method appears twice, once with its BIC-selected
parameter and once with its AIC-selected parameter.  Bold: lowest
BIC and lowest AIC across all rows.}
\label{tab:realdata_hr_danube}
\end{table}

\begin{figure}[ht]
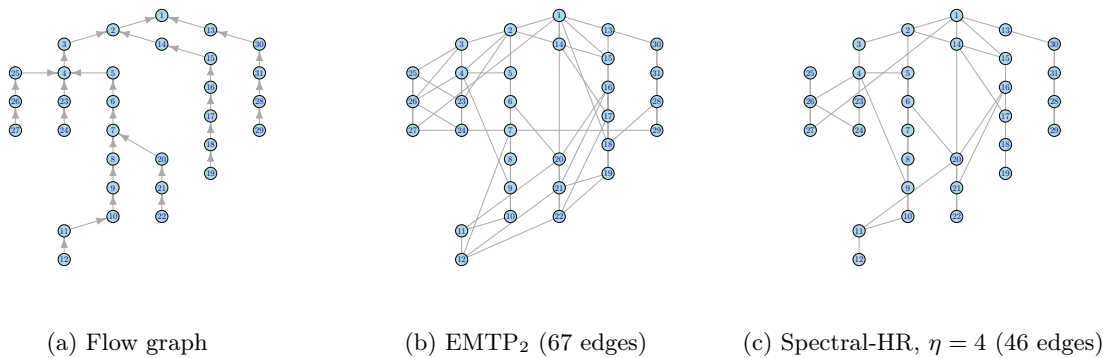

\centering
\begin{minipage}{0.32\linewidth}
\centering
\resizebox{\linewidth}{!}{\input{flowgraph_directed}}\\[-0.4em]
{\small (a) Flow graph}
\end{minipage}\hfill
\begin{minipage}{0.32\linewidth}
\centering
\resizebox{\linewidth}{!}{\input{emtp2_graph}}\\[-0.4em]
{\small (b) \EMTPtwo{} ($67$ edges)}
\end{minipage}\hfill
\begin{minipage}{0.32\linewidth}
\centering
\resizebox{\linewidth}{!}{\input{spectral_hr_bic_graph}}\\[-0.4em]
{\small (c) Spectral-HR, $\eta=4$ ($46$ edges)}
\end{minipage}
\caption{Danube river network: (a) the directed flow graph from
\citet{engelkehitz}; (b) the \EMTPtwo{}-MLE fit; and (c) the
Spectral-HR fit at the BIC-selected $\eta=4$.  Node positions and
numbering match across panels.}
\label{fig:danube_graphs}
\end{figure}

Spectral-HR's edge count grows monotonically over $\eta$ and stays
below \EMTPtwo{}'s $67$ edges across the entire grid.  eglearn-NS
admits a usable fit only for $\rho<0.132$, and past the
BIC-selected $\rho=0.032$ both criteria rise above the dense
\EMTPtwo{} baseline; eglearn-glasso attains its joint BIC/AIC
minimum at the second grid point ($\rho=0.044$, $171$ edges).  Full
sweeps are in \Cref{tab:danube_paths} (\Cref{sec:appendix_sweeps}),
where we observe that the eglearn performance decays for stronger
sparsity settings. Therefore, our impression for the Danube data set is that our method
behaves better if we are looking for a sparser model with reasonable performance.

\subsection*{Texas Cluster of U.S.\ airports \citep{HES2022}.}
The data are daily total flight delays at $d=29$ airports, accessed via
the \texttt{graphicalExtremes} R package, with the predefined
train/test date split of \citet{HES2022}
($n_{\mathrm{train}}=1{,}764$, $n_{\mathrm{test}}=1{,}839$ days). 
At
threshold $p=0.85$ this yields
$n^{\mathrm{thr}}_{\mathrm{tr}}=1{,}069$ training and
$n^{\mathrm{thr}}_{\mathrm{te}}=1{,}336$ held-out 
exceedances.
There is no
canonical extremal ground-truth graph for flight delays, so we evaluate
each fit by edge count, training H\"usler--Reiss BIC (the tuning criterion),
and held-out H\"usler--Reiss log-likelihood.  Tuning grids
($\rho\in[0.005,0.5]$ for eglearn, $\eta\in[1.1,100]$ for Spectral-HR)
are chosen wide enough that the BIC-selected parameter is interior on
the grid.

\Cref{tab:realdata_hr} reports the BIC-selected results.  Spectral-HR
gives the sparsest fit by a clear margin (95 edges, vs 130 for \EMTPtwo{}
and 211--227 for the eglearn variants).  \EMTPtwo{} attains the best
held-out test log-likelihood, with Spectral-HR a close second
($-23{,}104$ vs $-23{,}209$) at a $27\%$ reduction in edge count.  The
eglearn variants achieve the best training BIC but visibly overfit:
their held-out log-likelihoods ($-23{,}851$ and $-23{,}900$) are
substantially worse than both \EMTPtwo{} and Spectral-HR despite the
lower training BIC.

\begin{table}[ht]
\centering
\begin{tabular}{lrrr}
\toprule
Method & Edges & Train BIC & HR test log-lik.\ \\
\midrule
\EMTPtwo{}          & $130$         & $38{,}135.6$              & $\mathbf{-23{,}104.1}$ \\
eglearn-NS          & $211$         & $\mathbf{37{,}288.9}$     & $-23{,}899.8$ \\
eglearn-glasso      & $227$         & $37{,}369.6$              & $-23{,}851.3$ \\
Spectral-HR         & $\mathbf{95}$ & $37{,}967.5$              & $-23{,}208.8$ \\
\bottomrule
\end{tabular}
\caption{Texas Cluster of U.S. airports, daily total flight delays
(H\"{u}sler--Reiss, $d=29$, training/test split sizes
$n_{\mathrm{train}}=1{,}764$, $n_{\mathrm{test}}=1{,}839$ from
\citet{HES2022}, threshold $p=0.85$; training/test 
exceedance counts $n^{\mathrm{thr}}_{\mathrm{tr}}=1{,}069$,
$n^{\mathrm{thr}}_{\mathrm{te}}=1{,}336$).  Tuning by H\"usler--Reiss BIC
on the training sample; selected parameters $\rho=0.032$
(eglearn-NS), $\rho=0.05$ (eglearn-glasso), $\eta=20$ (Spectral-HR).
Train BIC is the tuning criterion
$-2\ell^{\mathrm{HR}}(\widehat\Theta;\,Y^{\mathrm{tr}}_{\mathrm{thr}})
+ k(\widehat\Theta)\log n^{\mathrm{thr}}_{\mathrm{tr}}$; the last
column is the held-out H\"{u}sler--Reiss log-likelihood.  Bold: best on
each metric.}
\label{tab:realdata_hr}
\end{table}

\section{Discussion}\label{sec:conclusion}
We extended the post-estimation spectral-sparsification approach of
\citet{SpectralMTP2026} from the full-rank Gaussian \MTPtwo{} setting
to two families of graphical models with a graph Laplacian precision matrix,
the LCGGM and the \EMTPtwo{} HRGM. Given that both models are estimated
in a similar way, we can use the same algorithm in both settings.
In each setting we run BSS on the dense estimator and
refit on the sparsified support. A Bregman--Loewner identity
on $\one^\perp$ underpins the analysis of both settings.

The powerful spectral sparsification tools, together with the models we
consider, ensure that the sparsified matrices are still valid statistical
parameters that are spectrally close to the original ones. For HRGMs, we
provide additional guarantees.
We evaluate Spectral-LCGGM empirically in simulations, where
we observe that our method serves as a post-processing operation which reduces
the complexity of the model while keeping performance high, or even improving
it in some instances. We also observe that our proposed method seems
to be quite precise at structure learning in this setting.

Additionally, we evaluate Spectral-HR in two real data examples, where we
observe that in practice our method seems to offer a trade-off between
model fit and complexity: we do not necessarily improve the model fit,
but we manage to simplify it without significant loss in performance.

Regarding future work, we believe there is still much to explore around these methods.
First, we believe that the refit step could be modified to include constraints
beyond positivity. Currently, we refit by imposing the sparsified support under
a Laplacian constraint, but we would like to explore whether other refitting
options behave differently.
Additionally, one may consider other spectral sparsification methods for the sparsification step. Regarding the performance of our methods, we believe it is possible to obtain more detailed theory specific to Spectral-HR, and overall we would find it interesting to see whether support recovery results are achievable, since the methods seem to recover structure quite well in simulations. Finally, we would like to explore more settings for simulations and real data experiments to further support our methods with additional empirical evidence of their performance.

\section*{Acknowledgements}
The research of Aida Abiad is supported by NWO (Dutch Research Council) through the grant  VI.Vidi.213.085.

\bibliographystyle{plainnat}
\bibliography{references}

\appendix

\section{Proofs}\label{sec:proofs}

\subsection{Bregman--Loewner Lemmas}\label{sec:proofs_toolbox}

The proof of \Cref{thm:main_lcggm} (which, applied under the
substitution of \Cref{sec:hr_theory}, also yields the H\"usler--Reiss identities
\eqref{eq:bregman_hr}--\eqref{eq:loglik_hr_bound}) and the associated
refit corollaries reduce to two elementary lemmas about graph Laplacians
sharing the kernel $\spa(\one)$.  The full-rank versions of these
lemmas, for positive definite matrices, were proved in
\citet[Sec.~3]{SpectralMTP2026}.  The two statements below adapt
them to the degenerate-rank setting that arises here for both LCGGMs
and the H\"{u}sler--Reiss surrogate.

\begin{lemma}\label{lem:kl_loewner}
Let $K,\widetilde K$ be positive semidefinite matrices with common kernel
$\spa(\one)$, and suppose $(1-\varepsilon)K\NSD\widetilde K\NSD(1+\varepsilon)K$
on $\one^\perp$ for some $\varepsilon\in(0,1)$.  Then
$D_{\mathrm{KL}}^\perp(K\,\|\,\widetilde K)\ge 0$, and
\begin{equation}\label{eq:kl_loewner}
2D_{\mathrm{KL}}^\perp(K\,\|\,\widetilde K)
\;\le\;\frac{(d-1)\varepsilon^2}{2(1-\varepsilon)}.
\end{equation}
In particular, $2D_{\mathrm{KL}}^\perp\le(d-1)\varepsilon^2$ for
$\varepsilon\le\tfrac12$.
\end{lemma}

\begin{proof}
Working on $\one^\perp$, let $\mu_1,\ldots,\mu_{d-1}$ be the eigenvalues of
$\widetilde K K^+$.  The spectral approximation gives
$\mu_i\in[1-\varepsilon,1+\varepsilon]$.  The standard KL identity for
degenerate centered Gaussians on $\one^\perp$ is
$2D_{\mathrm{KL}}^\perp(K\,\|\,\widetilde K)=\sum_{i=1}^{d-1}(\mu_i-1-\log\mu_i)\ge 0$.
Writing $\mu_i=1+t_i$ with $|t_i|\le\varepsilon<1$ and applying the bound
$t-\log(1+t)\le t^2/(2(1-|t|))$ for $|t|<1$ (Taylor with integral remainder)
yields $\mu_i-1-\log\mu_i\le\varepsilon^2/(2(1-\varepsilon))$; summing gives
\eqref{eq:kl_loewner}.  The simpler bound $(d-1)\varepsilon^2$ follows for
$\varepsilon\le\tfrac12$ since $1/(2(1-\varepsilon))\le 1$.
\end{proof}

\begin{lemma}\label{lem:optr}
Let $K,\widetilde K$ be positive semidefinite matrices with common kernel
$\spa(\one)$ and $-\varepsilon K\NSD\widetilde K-K\NSD\varepsilon K$ on
$\one^\perp$.  Then for every symmetric matrix $A$ supported on $\one^\perp$,
\[
\bigl|\tr\!\bigl((\widetilde K-K)A\bigr)\bigr|
\;\le\;\varepsilon\,\bigl\|K^{1/2}AK^{1/2}\bigr\|_*,
\]
where $K^{1/2}$ is the symmetric square root of $K$ on $\one^\perp$.
\end{lemma}

\begin{proof}
Let $V\in\RR^{d\times(d-1)}$ have orthonormal columns spanning $\one^\perp$
and write $\bar X:=V^\top XV$ for any symmetric matrix $X$ supported on
$\one^\perp$.  Then $\tr(XY)=\tr(\bar X\bar Y)$ for such matrices, and
$\bar K\succ 0$.  The spectral approximation is equivalent to
$\|\bar K^{-1/2}(\bar{\widetilde K}-\bar K)\bar K^{-1/2}\|_{\mathrm{op}}\le\varepsilon$.
The operator--trace H\"older inequality gives
\[
\bigl|\tr\!\bigl((\widetilde K-K)A\bigr)\bigr|
=\bigl|\tr\!\bigl(\bar K^{-1/2}(\bar{\widetilde K}-\bar K)\bar K^{-1/2}\cdot\bar K^{1/2}\bar A\bar K^{1/2}\bigr)\bigr|
\le\varepsilon\,\bigl\|\bar K^{1/2}\bar A\bar K^{1/2}\bigr\|_*,
\]
and $\|\bar K^{1/2}\bar A\bar K^{1/2}\|_*=\|K^{1/2}AK^{1/2}\|_*$.
\end{proof}

\subsection{Proofs for Section~\ref{sec:lcggm}}\label{sec:proofs_lcggm}

\begin{lemma}[Closed form of the degenerate-Gaussian KL]\label{lem:kl_perp_identity}
Let $K,\widetilde K$ be positive semidefinite matrices with common
kernel $\spa(\one)$.  The KL divergence between the two centered
degenerate Gaussians on $\one^\perp$ with precisions $K$ and
$\widetilde K$ is
\begin{equation}\label{eq:kl_perp_identity}
D_{\mathrm{KL}}^\perp(K\,\|\,\widetilde K)
\;=\;\tfrac12\bigl[\tr(\widetilde K K^+)-(d-1)-\log\Det(\widetilde K K^+)\bigr]
\;=\;\tfrac12\sum_{i=1}^{d-1}(\mu_i-1-\log\mu_i),
\end{equation}
where $\mu_1,\dots,\mu_{d-1}$ are the eigenvalues of $\widetilde K K^+$
on $\one^\perp$ (equivalently, the generalized eigenvalues of $\widetilde K$
relative to $K$ on $\one^\perp$).  In particular,
$D_{\mathrm{KL}}^\perp(K\,\|\,\widetilde K)\ge 0$ with equality iff
$\widetilde K=K$ on $\one^\perp$.
\end{lemma}

\begin{proof}
On $\one^\perp$, both Gaussians are absolutely continuous with respect
to the $(d-1)$-dimensional Lebesgue measure on that hyperplane, with
densities
\[
f_K(y)=\frac{(\Det K)^{1/2}}{(2\pi)^{(d-1)/2}}\,
       \exp\!\bigl(-\tfrac12 y^\top K y\bigr),
\quad
f_{\widetilde K}(y)=\frac{(\Det \widetilde K)^{1/2}}{(2\pi)^{(d-1)/2}}\,
       \exp\!\bigl(-\tfrac12 y^\top \widetilde K y\bigr).
\]
Both densities live on the same measure space because $K$ and
$\widetilde K$ share the kernel $\spa(\one)$.  The log-density
difference equals
$\tfrac12\log\Det K-\tfrac12\log\Det\widetilde K-\tfrac12 y^\top(K-\widetilde K)y$.
Taking expectation under $Y\sim f_K$ with $\mathbb{E}_K[YY^\top]=K^+$,
\[
\mathbb{E}_K[Y^\top(K-\widetilde K)Y]
=\tr((K-\widetilde K)K^+)
=\tr(K K^+)-\tr(\widetilde K K^+)
=(d-1)-\tr(\widetilde K K^+),
\]
using $KK^+=P$, the projector onto $\one^\perp$ with $\tr(P)=d-1$.
Substituting,
\[
D_{\mathrm{KL}}^\perp(K\,\|\,\widetilde K)
=\tfrac12\bigl[\tr(\widetilde K K^+)-(d-1)+\log\Det K-\log\Det\widetilde K\bigr].
\]
Since $K$ and $\widetilde K$ share kernel $\spa(\one)$, both are
(pseudo-) invertible on $\one^\perp$ and the determinant product rule on
$\one^\perp$ gives
$\log\Det K-\log\Det\widetilde K=-\log\Det(\widetilde K K^+)$,
yielding the first equality in \eqref{eq:kl_perp_identity}.  For the
second, let $\mu_i$ be the eigenvalues of $\widetilde K K^+$ on
$\one^\perp$; then $\tr(\widetilde K K^+)=\sum_i\mu_i$,
$\log\Det(\widetilde K K^+)=\sum_i\log\mu_i$, and $d-1=\sum_i 1$, so the
three terms combine into $\sum_i(\mu_i-1-\log\mu_i)$.
Non-negativity, with equality iff every $\mu_i=1$, follows from
$\mu-1-\log\mu\ge 0$ for $\mu>0$.
\end{proof}

\begin{proof}[Proof of \Cref{thm:main_lcggm}]
By the KL identity for degenerate Gaussians on $\one^\perp$,
\[
2D_{\mathrm{KL}}^\perp(\widehat{K}\,\|\,\widetilde K)
=\tr(\widetilde K \widehat{K}^+)-(d-1)+\log\Det \widehat{K}-\log\Det\widetilde K,
\]
so $\log\Det\widetilde K-\log\Det \widehat{K}=\tr(\widetilde K \widehat{K}^+)-(d-1)-2D_{\mathrm{KL}}^\perp(\widehat{K}\,\|\,\widetilde K)$.
Since $\widehat{K}\widehat{K}^+=P$ is the projector onto $\one^\perp$, $\tr(\widehat{K}\widehat{K}^+)=d-1$ and
$\tr(\widetilde K \widehat{K}^+)-(d-1)=\tr((\widetilde K-\widehat{K})\widehat{K}^+)$.  Substituting into
$\ell(\widetilde K;T)-\ell(\widehat{K};T)=(\log\Det\widetilde K-\log\Det \widehat{K})-\tr((\widetilde K-\widehat{K})T)$
gives
\[
\ell(\widetilde K;T)-\ell(\widehat{K};T)
=-2D_{\mathrm{KL}}^\perp(\widehat{K}\,\|\,\widetilde K)+\tr\!\bigl((\widetilde K-\widehat{K})(\widehat{K}^+-T)\bigr),
\]
which is \eqref{eq:bregman_lcggm}.  The upper-side bound
\eqref{eq:loglik_lcggm_upper} follows by dropping the non-positive KL term
and applying \Cref{lem:optr} with $A=T-\widehat{K}^+$.  For
\eqref{eq:loglik_lcggm_bound}, combine \Cref{lem:kl_loewner}
($0\le 2D_{\mathrm{KL}}^\perp\le(d-1)\varepsilon^2/(2(1-\varepsilon))$ for
$\varepsilon\in(0,1)$) with \Cref{lem:optr}.
\end{proof}

\subsection{Proofs for Section~\ref{sec:hr}}\label{sec:proofs_hr}

The H\"usler--Reiss Bregman identity \eqref{eq:bregman_hr} and the bound
\eqref{eq:loglik_hr_bound} are immediate from \Cref{thm:main_lcggm}
under the substitution of \Cref{sec:hr_theory}.

\begin{proof}[Proof of \Cref{thm:variogram_preservation}]
On $\one^\perp$ both Laplacians are positive definite.  The spectral
approximation \eqref{eq:spec_hr} together with operator-monotonicity of the
inverse on positive definite operators give
$\Theta^+/(1+\varepsilon)\NSD\widetilde\Theta^+\NSD\Theta^+/(1-\varepsilon)$
on $\one^\perp$.  For each pair $i\neq j$, the vector
$b_{ij}:=e_i-e_j$ lies in $\one^\perp$, so
$\widetilde\Gamma_{ij}=b_{ij}^\top\widetilde\Theta^+b_{ij}$ inherits the
two-sided bound, giving \eqref{eq:variogram_preservation}.  The same
argument applies to any pair, so all variogram entries are preserved multiplicatively.
\end{proof}

\begin{proof}[Proof of \Cref{cor:loglambda_pointwise}]
Let $v(y):=\binom{y}{1}\in\RR^{d+1}$.  By the H\"{u}sler--Reiss density
representation \eqref{eq:hr_density_gamma},
\[
\log\lambda_{\widetilde\Theta}(y)-\log\lambda_\Theta(y)
=-\tfrac12\bigl(\log|\det\widetilde C|-\log|\det C|\bigr)
  -\tfrac12\,v(y)^\top(\widetilde C^{-1}-C^{-1})\,v(y).
\]

\emph{Step 1 (Cayley--Menger perturbation).}  The Cayley--Menger blocks
yield $\widetilde C-C=\bigl(\begin{smallmatrix}-(\widetilde\Gamma-\Gamma)/2 & 0\\ 0 & 0\end{smallmatrix}\bigr)$,
so $\|\widetilde C-C\|_2=\tfrac{1}{2}\|\widetilde\Gamma-\Gamma\|_2$.
The two-sided bound \eqref{eq:variogram_preservation} of
\Cref{thm:variogram_preservation} reads
$\Gamma_{ij}/(1+\varepsilon)\le\widetilde\Gamma_{ij}\le\Gamma_{ij}/(1-\varepsilon)$,
which translates to the entrywise bound
$|\widetilde\Gamma_{ij}-\Gamma_{ij}|\le\tfrac{\varepsilon}{1-\varepsilon}\Gamma_{ij}$
(the upper deviation $\varepsilon/(1-\varepsilon)$ dominates the lower
$\varepsilon/(1+\varepsilon)$).  Since $\Gamma$ has nonnegative entries,
monotonicity of the spectral radius on nonnegative matrices (Wielandt's
theorem) gives
$\rho(|\widetilde\Gamma-\Gamma|)\le\tfrac{\varepsilon}{1-\varepsilon}\rho(\Gamma)$,
and the symmetry of both matrices together with $\|A\|_2\le\rho(|A|)$
for symmetric $A$ gives
$\|\widetilde\Gamma-\Gamma\|_2\le\rho(|\widetilde\Gamma-\Gamma|)\le\tfrac{\varepsilon}{1-\varepsilon}\|\Gamma\|_2$.
Hence $\|\widetilde C-C\|_2\le\eta_C$.  Using $\|\Gamma\|_2$ here,
instead of the entrywise Frobenius bound
$\|\widetilde\Gamma-\Gamma\|_F\le\tfrac{\varepsilon}{1-\varepsilon}\|\Gamma\|_F$
that an entry-by-entry summation would give, saves up to a factor of
$\sqrt{d}$ for generic variograms.

\emph{Step 2 (log-determinant and inverse perturbation).}  Write
$M_\Theta:=\|C^{-1}\|_{\mathrm{op}}=1/s_{\min}(C)$, so
$s_i(C)\ge s_{\min}(C)=1/M_\Theta$ for every $i$.  Assuming
$\eta_C M_\Theta<1$, Weyl's inequality for singular values
($|s_i(\widetilde C)-s_i(C)|\le\|\widetilde C-C\|_2\le\eta_C$) gives
\[
s_i(\widetilde C)\ge s_i(C)-\eta_C\ge 1/M_\Theta-\eta_C
=(1-\eta_C M_\Theta)/M_\Theta>0
\quad\text{for every }i,
\]
and, writing $t_i:=s_i(\widetilde C)/s_i(C)-1$,
$|t_i|\le\eta_C/s_i(C)\le\eta_C M_\Theta<1$.  Using
$|\log(1+t)|\le|t|/(1-|t|)$ for $|t|<1$ (immediate from
$\log(1+t)=\int_0^1 t/(1+\tau t)\,d\tau$), each term satisfies
$|\log(s_i(\widetilde C)/s_i(C))|=|\log(1+t_i)|\le|t_i|/(1-|t_i|)\le\eta_C M_\Theta/(1-\eta_C M_\Theta)$.
Summing over the $d+1$ singular values of $C$,
\begin{align*}
\bigl|\log|\det\widetilde C|-\log|\det C|\bigr|
&=\Bigl|\sum_{i=1}^{d+1}\log\bigl(s_i(\widetilde C)/s_i(C)\bigr)\Bigr|\\
&\le\sum_{i=1}^{d+1}|\log(1+t_i)|
\;\le\;(d+1)\,\eta_C M_\Theta/(1-\eta_C M_\Theta),
\end{align*}
where the first equality uses $|\det A|=\prod_i s_i(A)$.
For the inverse perturbation, the resolvent identity
$\widetilde C^{-1}-C^{-1}=C^{-1}(C-\widetilde C)\widetilde C^{-1}$ together
with $\|C^{-1}\|_2=M_\Theta$ and
$\|\widetilde C^{-1}\|_2=1/s_{\min}(\widetilde C)\le M_\Theta/(1-\eta_C M_\Theta)$
gives
\[
\|\widetilde C^{-1}-C^{-1}\|_2\le M_\Theta\cdot\eta_C\cdot M_\Theta/(1-\eta_C M_\Theta)
=\eta_C M_\Theta^2/(1-\eta_C M_\Theta),
\]
whence $|v^\top(\widetilde C^{-1}-C^{-1})v|\le\|v\|_2^2\,\|\widetilde C^{-1}-C^{-1}\|_2
=(\|y\|_2^2+1)\,\eta_C M_\Theta^2/(1-\eta_C M_\Theta)$
(using $\|v(y)\|_2^2=\|y\|_2^2+1$).
Combining the two bounds yields \eqref{eq:loglambda_pointwise}.  The
block bound $M_\Theta\le\|\Theta\|_{\mathrm{op}}+\|p\|_2+|\sigma^2|$ in
\eqref{eq:M_C_def} follows by triangle inequality on the three-block
decomposition $\mathrm{CM}(\Gamma)^{-1}=\mathrm{diag}(\Theta,0)+\bigl(\begin{smallmatrix}0&p\\p^\top&0\end{smallmatrix}\bigr)+\mathrm{diag}(0,\sigma^2)$,
whose summands have operator norms $\|\Theta\|_{\mathrm{op}}$, $\|p\|_2$,
and $|\sigma^2|$.
The sample-mean bound \eqref{eq:loglambda_sample} follows by averaging the
pointwise bound over $Y_1,\ldots,Y_n$.
\end{proof}

\section{Tuning-parameter sweeps for the Danube experiment}\label{sec:appendix_sweeps}

\begin{table}[ht]
\centering
\footnotesize
\begin{minipage}[t]{0.32\linewidth}
\centering
(a) Spectral-HR ($\eta$)\\[0.4em]
\resizebox{\linewidth}{!}{%
\begin{tabular}{rrrr}
\toprule
$\eta$ & Edges & BIC & AIC \\
\midrule
$1.10$   & $31$ & $1{,}521.17$            & $1{,}435.52$ \\
$1.20$   & $33$ & $1{,}364.65$            & $1{,}273.51$ \\
$1.30$   & $34$ & $1{,}342.50$            & $1{,}248.57$ \\
$1.50$   & $35$ & $1{,}304.76$            & $1{,}208.03$ \\
$2.00$   & $38$ & $1{,}285.98$            & $1{,}181.02$ \\
$2.50$   & $39$ & $1{,}281.27$            & $1{,}173.57$ \\
$3.00$   & $42$ & $1{,}284.30$            & $1{,}168.24$ \\
$4.00$   & $46$ & $\mathbf{1{,}264.47}$   & $1{,}137.42$ \\
$5.00$   & $46$ & $1{,}264.51$            & $1{,}137.36$ \\
$6.00$   & $48$ & $1{,}268.55$            & $1{,}135.94$ \\
$8.00$   & $48$ & $1{,}268.56$            & $1{,}135.97$ \\
$10.00$  & $48$ & $1{,}268.54$            & $1{,}135.96$ \\
$15.00$  & $50$ & $1{,}265.87$            & $1{,}127.77$ \\
$20.00$  & $52$ & $1{,}268.46$            & $1{,}124.90$ \\
$25.00$  & $53$ & $1{,}268.75$            & $\mathbf{1{,}122.36}$ \\
$30.00$  & $53$ & $1{,}268.76$            & $1{,}122.35$ \\
$40.00$  & $55$ & $1{,}277.55$            & $1{,}125.62$ \\
$50.00$  & $56$ & $1{,}281.89$            & $1{,}127.23$ \\
$75.00$  & $57$ & $1{,}287.81$            & $1{,}130.40$ \\
$100.00$ & $60$ & $1{,}302.26$            & $1{,}136.53$ \\
\bottomrule
\end{tabular}}
\end{minipage}\hfill
\begin{minipage}[t]{0.32\linewidth}
\centering
(b) eglearn-NS ($\rho$)\\[0.4em]
\resizebox{\linewidth}{!}{%
\begin{tabular}{rrrr}
\toprule
$\rho$ & Edges & BIC & AIC \\
\midrule
$0.005$  & $364$           & $2{,}012.63$          & $1{,}007.19$ \\
$0.014$  & $184$           & $1{,}342.91$          & $\phantom{1{,}}834.65$ \\
$0.023$  & $130$           & $1{,}179.36$          & $\phantom{1{,}}\mathbf{820.27}$ \\
$0.032$  & $\phantom{1}99$ & $\mathbf{1{,}176.96}$ & $\phantom{1{,}}903.52$ \\
$0.041$  & $\phantom{1}89$ & $1{,}281.88$          & $1{,}036.02$ \\
$0.050$  & $\phantom{1}82$ & $1{,}309.73$          & $1{,}083.25$ \\
$0.091$  & $\phantom{1}56$ & $1{,}519.46$          & $1{,}364.76$ \\
$\ge 0.132$ & ---          & ---                   & --- \\
\bottomrule
\end{tabular}}
\end{minipage}\hfill
\begin{minipage}[t]{0.32\linewidth}
\centering
(c) eglearn-glasso ($\rho$)\\[0.4em]
\resizebox{\linewidth}{!}{%
\begin{tabular}{rrrr}
\toprule
$\rho$ & Edges & BIC & AIC \\
\midrule
$0.005$ & $291$ & $1{,}694.26$          & $\phantom{1{,}}890.46$ \\
$0.044$ & $171$ & $\mathbf{1{,}342.77}$ & $\phantom{1{,}}\mathbf{870.42}$ \\
$0.083$ & $151$ & $1{,}419.73$          & $1{,}002.66$ \\
$0.123$ & $146$ & $1{,}531.34$          & $1{,}128.08$ \\
$0.162$ & $160$ & $1{,}538.49$          & $1{,}096.53$ \\
$0.201$ & $167$ & $1{,}556.78$          & $1{,}095.49$ \\
$0.240$ & $174$ & $1{,}577.94$          & $1{,}097.32$ \\
$0.279$ & $171$ & $1{,}571.67$          & $1{,}099.38$ \\
$0.319$ & $170$ & $1{,}543.69$          & $1{,}074.12$ \\
$0.358$ & $170$ & $1{,}558.03$          & $1{,}088.46$ \\
$0.397$ & $162$ & $1{,}566.80$          & $1{,}119.37$ \\
$0.436$ & $158$ & $1{,}544.02$          & $1{,}107.65$ \\
$0.476$ & $150$ & $1{,}567.54$          & $1{,}153.21$ \\
$0.515$ & $134$ & $1{,}561.81$          & $1{,}191.68$ \\
$0.554$ & $128$ & $1{,}709.41$          & $1{,}355.88$ \\
$0.593$ & $117$ & $1{,}707.47$          & $1{,}384.33$ \\
$\ge 0.632$ & ---   & ---                   & --- \\
\bottomrule
\end{tabular}}
\end{minipage}
\caption{Full tuning-parameter sweeps on the Danube data for the
three tunable methods: edge count, BIC and AIC \eqref{eq:bic_aic_hr}
at every grid point.  (a) Spectral-HR over the BSS parameter
$\eta$, (b) eglearn-NS over the $\ell_1$ parameter $\rho$, (c)
eglearn-glasso over $\rho$.  Bold marks the parameter selected by
each criterion.  Rows listed as ``---'' (in (b) for $\rho\geq 0.132$
and in (c) for $\rho\geq 0.632$) correspond to grid points where the
estimated graph is disconnected and $\widehat\Gamma$ cannot be
completed.}
\label{tab:danube_paths}
\end{table}

\end{document}